\documentclass[final,1p,times,authoryear]{elsarticle}

\usepackage{amsmath}
\usepackage{amsthm}
\usepackage{amssymb}
\usepackage{amsfonts,url}

\newtheorem{theorem}{Theorem}
\newtheorem{lemma}[theorem]{Lemma}
\newtheorem{corollary}[theorem]{Corollary}

\newtheorem{proposition}[theorem]{Proposition}

\newtheorem{example}[theorem]{Example}

\usepackage[T1]{fontenc}

\newcommand{\calV}{\mathcal{V}}

\usepackage{graphicx}
\def\chk#1{#1^{\smash{\scalebox{.7}[1.4]{\rotatebox{90}{\guilsinglleft}}}}}

\def\zarA{{\mathscr{A}}}
\def\zarC{{\mathscr{C}}}

\def\lag{\text{lag}}

\usepackage{algorithmicx}
\usepackage{algorithm}
\usepackage[noend]{algpseudocode}
\makeatletter
\def\BState{\State\hskip-\ALG@thistlm}
\makeatother

\usepackage[authoryear]{natbib}

\usepackage[linkcolor=red, urlcolor=blue, citecolor=red, colorlinks=true]{hyperref}

\def\sdp{{\text{(SDP)}}}
\def\NN{{\mathbb{N}}} % natural numbers
\def\QQ{{\mathbb{Q}}} % rational numbers
\def\RR{{\mathbb{R}}} % real numbers
\def\CC{{\mathbb{C}}} % complex numbers
\def\ZZ{{\mathbb{Z}}} % integer numbers

\def\rank{{\rm rank}}
\def\deg{{\rm deg}}
\def\mymid{{\,\,:\,\,}}

\def\bigO{\ensuremath{{\mathcal{O}}}}
\def\cc{{{C}}}
\def\spec{\mathscr{S}}
\usepackage{amssymb,amsmath,mathrsfs}
\newcommand{\calD}{\mathscr{D}}

\newcommand{\calR}{\mathcal{R}}
\newcommand{\calF}{\mathcal{F}}

\newcommand{\realset}[1]{{\mathsf{Z}_\RR(#1)}}
\newcommand{\zeroset}[1]{{\mathsf{Z}(#1)}}
\newcommand{\ideal}[1]{{\mathsf{I}(#1)}}
\def\jac{{D\,}} % jacobian matrix

\def\crit{{\rm crit}\,}

\def\Id{{\rm I}}
\def\ss{{\mathbb{S}}}

\usepackage{comment}

\begin{document}

\begin{frontmatter}

\title{Solving rank-constrained semidefinite programs in exact arithmetic}

\author{Simone Naldi}
\address{Technische Universit\"at Dortmund \\ Fakult\"at f\"ur Mathematik \\ Vogelpothsweg 87, 44227 Dortmund}
\ead{simone.naldi@tu-dortmund.de}
\ead[]{http://www.mathematik.tu-dortmund.de/sites/simone-naldi}

\begin{abstract}
We consider the problem of minimizing a linear function over an affine section
of the cone of positive semidefinite matrices, with the additional constraint
that the feasible matrix has prescribed rank. When the rank constraint is active,
this is a non-convex optimization problem, otherwise it is a semidefinite program.
Both find numerous applications especially in systems control theory and combinatorial
optimization, but even in more general contexts such as polynomial optimization or
real algebra. While numerical algorithms exist for solving this problem, such as
interior-point or Newton-like algorithms, in this paper we propose an approach based
on symbolic computation. We design an exact algorithm for solving rank-constrained
semidefinite programs, whose complexity is essentially quadratic on natural degree
bounds associated to the given optimization problem: for subfamilies of the problem
where the size of the feasible matrix, or the dimension of the affine section,
is fixed, the algorithm is polynomial time. The algorithm works under assumptions
on the input data: we prove that these assumptions are generically satisfied. We
implement it in Maple and discuss practical experiments.
\end{abstract}

\end{frontmatter}

\section{Introduction}
\label{sec-intro}

\subsection{Problem statement}
\label{ssec-intro-prob}
Let $x=(x_1,\ldots,x_n)$ denote a vector of unknowns. We consider the standard semidefinite
programming (SDP) problem with additional rank constraints, as follows:
\begin{equation}
\label{sdpr}
\begin{aligned}
\sdp_r \,\,\,\,\,\,\,\,\,\,\,\,\, \inf_{x \in \RR^n} & \,\,\, \ell_c(x) \\
 s.t. & \,\,\,\,A(x) \succeq 0 \\
                       & \,\,\,\,\rank\,A(x) \leq r 
\end{aligned}
\end{equation}
In Problem \eqref{sdpr}, $\ell_c(x)=c^Tx$, $c \in \QQ^n$, $A(x)=A_0+x_1A_1+\cdots+x_nA_n$
is a symmetric linear matrix with
$A_i \in \ss_m(\QQ)$ (the set of symmetric matrices of size $m$ with entries in $\QQ$),
and $r$ is an integer, $0 \leq r \leq m$. The formula $A(x) \succeq 0$ means that $A(x)$
is positive semidefinite ({\it i.e.}, all its eigenvalues are nonnegative) and is called
a linear matrix inequality (LMI). Remark that for $r=m$ this is the standard semidefinite
programming problem since the rank constraint is inactive. Moreover, when $c=0$ ({\it i.e.},
$c$ is the zero vector), $\sdp_r$ is a rank-constrained LMI. In the whole paper, we refer
to $\sdp_r$ in Problem \eqref{sdpr} as a rank-constrained semidefinite program with parameters
$(m,n,r)$.
The set
\[
\spec = \left\{x \in \RR^n \mymid A(x) \succeq 0\right\},
\]
namely the feasible set of $\sdp_m$, is called a spectrahedron by the convex algebraic geometry
community, or equivalently LMI-set. It is a convex basic semialgebraic set. Conversely, for $r<m$,
$\sdp_r$ is no more a convex optimization problem, in general. Indeed, denoted by
\[
\calD_p = \left\{x \in \CC^n \mymid \rank\,A(x) \leq p\right\}
\]
the complex determinantal variety associated to $A(x)$ of maximal rank $p$, the feasible set of $\sdp_r$ is
exactly $\spec \cap \calD_r \cap \RR^n = \spec \cap \calD_r$. This is typically non-convex.

The purpose of this paper is to design an exact algorithm for solving problem $\sdp_r$.

\subsection{Contribution}
\label{ssec-intro-contr}

We suppose that the input data is defined over the rational numbers, namely $(c,A_0,A_1,\ldots,A_n)
\in \QQ^n \times (\ss_{m}(\QQ))^{n+1}$. By exact, we mean that, the output of the algorithm is either
an empty list, or a finite set $S$ encoded by a rational parametrization as in \cite{Rouillier99}.
This is the exact algebraic representation encoded by a vector $(q,q_0,q_1,\ldots,q_n) \subset
\QQ[t]$ of univariate polynomials, such that $q_0,q$ are coprime and:
\begin{equation}
\label{ratrep}
S = \left\{ \left( \frac{q_1(t)}{q_0(t)}, \ldots, \frac{q_n(t)}{q_0(t)}\right) \mymid q(t)=0 \right\}.
\end{equation}
When $S$ is not empty, the degree of $q$ is the algebraic degree of every element in $S$.
When the output is not the empty list, the set $S$ which is returned contains at least
one minimizer $x^*$ of $\sdp_r$.
Under general assumptions on input data, which are highlighted and discussed below, the strategy
to reach our main goal is twofold:
\begin{itemize}
\item
we prove that the {\it semialgebraic} optimization problem $\sdp_r$ can be reduced to
a ({\it finite}) {\it sequence} of {\it algebraic} optimization problems, that is, whose feasible set is
real algebraic;
\item
we design {\it exact algorithms} for solving the reduced algebraic optimization problems.
\end{itemize}
Once a rational parametrization $(q,q_0,q_1,\ldots,q_n)$ of $S$ is known, the coordinates of a minimizer
can be approximated by intervals of (arbitrary length) of rational
numbers, by isolating the real solutions of the univariate equation $q(t)=0$. The complexity of
the real root isolation problem is quadratic in the degree of $q$ and linear in the total bitsize
of its coefficients; for more information, {\it cf.} \cite{PanTsi}.

Once the output is returned, one can compute the list of minimizers by
sorting the set $S$ with respect to the value of the objective function $\ell_c(x)$, and
deleting the solutions lying out of the feasible set $\spec \cap \calD_r$: hence, our goal
is also to give a bound for the maximal size of the output set $S$, namely, on the
degree of $q$.

\subsection{Motivations}
\label{ssec-intro-motiv}
Several problems in optimization are naturally modeled
by (rank-constrained) semidefinite programming, SDP for short, see {\it e.g.} \cite{anjlas},
\cite{vandenberghe1996semidefinite} or \cite{ben2001lectures}.
Given $f,f_1,\ldots,f_s \in \RR[x]$, the general polynomial optimization
problem
\begin{equation}
\label{genpop}
\begin{aligned}
f^* = \inf_{x \in \RR^n} & \,\,\, f(x) \\
 s.t. & \,\,\,\, f_1(x) \geq 0, \ldots f_s(x) \geq 0 
\end{aligned}
\end{equation}
reduces to a sequence of semidefinite programs of increasing size, see {\it e.g.}
\cite{lasserre01globaloptimization} and \cite{parrilo}. Since this sequence is almost always
finite by \cite{nie14}, lots of efforts have been made in order to develop efficient algorithms for SDP.
Moreover, LMI and SDP conditions frequently appear in systems control theory {\it cf.} \cite{boyd1994linear}.
Finding low-rank positive semidefinite matrices also concerns the completion
problem for some classes of matrices in combinatorics \cite{LNV}.
Finally, an independent application of SDP-based techniques, but highly related to the
polynomial optimization problem, is that of checking nonnegativity of multivariate polynomials.
Indeed, deciding whether a given $f \in \RR[u_1,\ldots,u_k]$ is a SOS (sum of squares) of at most
$r$ polynomials (hence, nonnegative) is equivalent to a rank-constrained semidefinite program
(see Section \ref{sosdec} and, {\it e.g.}, \cite{PowersWor}). Keeping track of the length of
a SOS decomposition, or just deciding whether such a decomposition exists, is crucial in different
contexts, {\it cf.} \cite{blekherman2016low}.

\subsection{Previous work}
\label{ssec-intro-state}

The ellipsoid method in \cite{ellipsoid} translates into an iterative algorithm for solving
general convex optimization problem. The number number of its iterations is polynomial in the
input size (measured by the size $m$ of the matrix and by the number $n$ of variables) with
fixed precision, see {\it e.g.} \cite{anjlas}, but this algorithm is known to be inefficient
in practice. On the other hand, the extension of Karmakar's interior-point method beyond linear
programming by \cite{nestnem} yields efficient algorithms for computing floating point
approximations of a solution, implemented in several solvers such as SeDuMi, SOSTOOLS {\it etc}.

However, these algorithms cannot, in general, manage additional determinantal conditions or
non-convexity. Moreover, SDP relaxations of hard combinatorial optimization problems (as the
MAX-CUT, see \cite{goemans}) usually discard such algebraic constraints, since they break
desirable convexity properties. Moreover, interior-point algorithms cannot certify the emptiness
of the feasible set or the rank of the optimal solution, and can often suffer of numerical
round-off errors. Remark that if the standard SDP problem $\sdp_m$ has a solution $x^*$ of
rank $r$, then $x^*$ is also a solution of the non-convex problem $\sdp_r$ (the viceversa is
false, in general). Finally, one cannot extract information about the algebraic degree
\cite{stu} of the solution with numerical methods. The output of the algorithm designed in
this paper allows to recover important information about the solution, namely the algebraic
degree of the entries of the optimal matrix $A(x^*)$ and its rank.

In \cite{OHM}, Newton-like ``tangent and lift'' and projection methods for approximating a
point in the intersection of a linear space and a manifold are proposed: the authors use this
approach for solving rank constrained LMI but, in general, without guarantees of convergence,
and with the request of a starting feasible point. In \cite{HNS2015c} an exact algorithm for
LMI has been proposed. This algorithm, implemented in the Maple library SPECTRA \cite{spectra},
has a runtime essentially quadratic on a multilinear B\'ezout bound on the output degree,
and polynomial in $n$ (resp. in $m$) when $m$ (resp. $n$) is fixed. This last property is
shared with the algorithm in \cite{porkolab1997complexity}, which, however, cannot be used in
practice, since it crucially relies on quantifier elimination techniques. The algorithm in
\cite{greuet2} is also exact, but cannot manage semialgebraic constraints and has regularity
assumptions on the input, which are not satisfied in our case. The related problem of computing
witness points on determinantal algebraic sets has been addressed and solved in
\cite{HNS2014,HNS2015a}.

Our contribution builds on the approach of \cite{HNS2015c}, based on
the lifted representation of determinantal sets $\calD_r$ via incidence varieties, which is
recalled and adapted to our situation in Section \ref{incvar}. However, the geometric results
in Sections \ref{ssec:critpoints} and \ref{sec-fromsemialg} are crucial to allow to extend this
method to the rank-constrained SDP problem.

\subsection{Outline of main results}
\label{ssec-intro-mainresults}

We consider the rank-constrained semidefinite programming problem \eqref{sdpr}, encoded by rational
data $(c,A) \in \QQ^n \times \ss_m^{n+1}(\QQ)$, and by the integer $r$ bounding the rank of an optimal
solution. Our paper can be divided into two parts.

In the first part (Sections \ref{sec-prelim} and \ref{sec-fromsemialg}) we prove geometrical properties
of problem $\sdp_r$. 
In Section \ref{incvar}, we represent the algebraic sets $\calD_p, p=0, \ldots, r$, as projections
of incidence varieties defined by bilinear equations, that are generically smooth and equidimensional
(Proposition \ref{smoothness}). The solutions of $\sdp_r$ are also
local minimizers of $\ell_c$ on $\calD_p \cap \RR^n$ (this is proved in Theorem \ref{reduction}) and
are obtained as the projection of critical points of the same map restricted to the incidence varieties
(Lemma \ref{lemmaProjCrit}), which are finitely many (Proposition \ref{finiteness}). As an outcome,
we prove that a {\it generic} rank-constrained semidefinite program admits finitely many minimizers
(Corollary \ref{finitesolsdpr}).

The second part hosts the formal description of an algorithm for solving $\sdp_r$ (Section \ref{algorithm})
and its correctness (Theorem \ref{correctness}). A complexity analysis is then performed in Section
\ref{complexityanalysissection}, with explicit bounds on the size of the output set $S$ ({\it cf.}
\eqref{ratrep}) computed in Proposition \ref{degbound}. We finally discuss the results of numerical
tests performed via a first implementation of our algorithm in Section \ref{exper}.

This revised and extended version of the paper \cite{naldiIssac2016} published in the Proceedings
of ISSAC 2016, contains examples explaining our methodology and an extended experimental section,
showing results of our tests performed via the Maple library {\sc spectra}, {\it cf.} \cite{spectra}.

\section{Preliminaries}
\label{sec-prelim}

\subsection{General notation}
\label{ssec-intro-notation}
If $f = \{f_1, \ldots, f_s\} \subset \QQ[x]$, we denote by $\zeroset{f}$ the set of
complex solutions of $f_1=0, \ldots, f_s=0$, called a complex algebraic set. We also consider
real solutions of polynomial equations, that is the real algebraic set $\realset{f}=\zeroset{f} \cap \RR^n$.
If $S \subset \CC^n$, the ideal of polynomials
vanishing on $S$ is denoted by $\ideal{S}$. An ideal $I \subset \RR[x]$ is called radical
if it equals its radical $\sqrt{I}=\{f \in \RR[x] \mymid \exists\,s \in \NN, f^s \in I\}$.
An ideal of type $\ideal{S}$ is always a radical ideal.
By Hilbert's Nullstellensatz, one has $\ideal{\zeroset{I}} = \sqrt{I}$. The Jacobian matrix
of partial derivatives of $\{f_1, \ldots, f_s\}$ is denoted by $\jac f = (\frac{\partial f_i}{\partial x_j})_{i,j}$.

An algebraic set $V \subset \CC^n$ is called irreducible if it is not the union of two
proper algebraic subsets; otherwise it is the finite union of irreducible algebraic sets
$V=V_1 \cup \cdots \cup V_s$, called the irreducible components. The dimension of $V$ is the
Krull dimension of its coordinate ring $\CC[x]/\ideal{V}$. If the $V_i$ in the previous
decomposition have the same dimension $d$, then $V$ is equidimensional of dimension $d$.
Let $V \subset \CC^n$ be equidimensional of co-dimension $c$, and let $\ideal{V}=\langle
f_1, \ldots, f_s \rangle$. We say that $V$ is smooth if its singular locus, that is the
algebraic set defined by $f=(f_1, \ldots, f_s)$ and by the $c \times c$ minors of
$\jac f$, is empty. A set ${\cal E} = {\zeroset{I}} \setminus \zeroset{J}$ is called
locally closed, and its dimension is the dimension of its Zariski closure $\zeroset{\ideal{\cal E}}$.

If $V$ is equidimensional and smooth, and if $g \colon \CC^n \to \CC^m$ is an algebraic map, the critical
points of the restriction of $g$ to $V$ are denoted by $\crit(g,V)$, and defined by
$f=(f_1, \ldots, f_s)$ and by the $c+m$ minors of $\jac(f,g)$. Equivalently, a point $x \in V$
is critical for $g$ on $V$ if and only if the differential map $dg \colon T_xV \to \CC^m$ is
not surjective (where $T_xV$ is the Zariski tangent space of $V$ at $x$, {\it cf.} \cite[Sec.\,2.1.2]{Shafarevich77}).
The elements of $g(\crit(g,V))$ are the critical values, and the elements of $\CC^m \setminus g(\crit(g,V))$
are the regular values of the restriction of $g$ to $V$.

Let $S \subset \RR^n$ be any set, and let $f \colon \RR^n \to \RR$ be a continuous function with respect
to the Euclidean topology of $\RR^n$ and $\RR$. A point $x^* \in S$ is a local minimizer of $f$ on $S$,
if there exists an Euclidean open set $U \subset \RR^n$ such that $x^* \in U$ and $f(x^*) \leq f(x)$ for
every $x \in U \cap S$. A point $x^* \in S$ is a minimizer of $f$ on $S$ if $f(x^*) \leq f(x)$ for
every $x \in S$. In particular, if $\cc \subset S$ is a connected component of $S$, every minimizer
of $f$ on $\cc$ is a local minimizer of $f$ on $S$.

We finally recall the notation introduced previously. We consider $m \times m$ symmetric matrices
$A_0,A_1,\ldots,A_n \in \ss_{m}(\QQ)$, and a linear matrix $A(x) = A_0+x_1A_1+\cdots+x_nA_n$. The convex
set $\spec = \{x \in \RR^n \mymid A(x) \succeq 0\}$ is called a spectrahedron. 
The integer $r \in \NN$ will denote the maximal admissible rank in Problem \eqref{sdpr}.
Given an integer
$p \in \NN$, with $0 \leq p \leq r$, we denote by $\calD_p = \{x \in \CC^n \mymid \rank\,A(x) \leq p\}$
the determinantal variety of {maximal rank} $p$ generated by $A(x)$. 

%We conclude this section by stating a global assumption on the input data.
%The minimizers of $\sdp_r$ are local minima of the linear function $\ell_c(x)$, corresponding,
%geometrically, to critical points of the projection map $\pi_c \colon \RR^n \to \myspan{c}$
%over the line $\myspan{c} = \{\lambda c \mymid \lambda \in \RR\} \subset \RR^n$ generated
%by $c$. In particular, one would expect for such a set to be finite. Throughout the paper,
%we assume the following property for the set of solutions of $\sdp_r$.
%\begin{assumption}
%\label{ass-finite}
%For $0 \leq r \leq m$, the set of minimizers of $\sdp_r$ is empty or finite.
%\end{assumption}
%Remark that Assumption \ref{ass-finite} does not exclude, a priori, non-compact cases,
%or degenerate ones (such as feasible sets with empty interior). This is important since these
%last two situations often occur when modeling problems from the applications with
%SDP. Moreover, Assumption \ref{ass-finite} implies, by convexity, that $\sdp_m$ has a unique
%solution.
%%We prove that Assumption  \ref{ass-finite}  holds generically.
%%\begin{proposition}
%%There is a non-empty Zariski open set $\zarO \subset \CC^n \times \ss_m^{n+1}(\CC)$ such that,
%%for $(c,A) \in \zarO \cap \QQ^n \times \ss_m^{n+1}(\QQ)$, the following holds. For $r=0,\ldots,m$,
%%either $\calF_r(A,c)$ is empty, or it is finite.
%%\end{proposition}
%In \cite{Shapiro} the uniqueness of the solution for generic semidefinite programs (that is,
%the fact that Assumption \ref{ass-finite} holds generically) is proved.

\subsection{Representation via incidence varieties}
\label{incvar}
The algebraic set $\calD_p$ will not be represented as the vanishing locus of the $(p+1) \times
(p+1)$ minors of $A(x)$, mainly by two reasons. The first is that computing determinants is
a difficult task. Even if this first issue could be avoided by some precomputation, the
singularities of determinantal varieties appear generically. We are going to represent
$\calD_p$ as the projection of a more regular algebraic set, reviewing a classical construction.

Let $V$ be a vector space of dimension $d$ and let $\mathbb{G}(e,d)$ be the Grassmannian
of linear subspaces of dimension $e$ of $V$, with $e \leq d$. Fixed a basis of $V$, a point
$L = \text{span}(v_1,\ldots,v_e) \in \mathbb{G}(e,d)$ is represented by the $d \times e$
matrix whose columns are $v_1,\ldots,v_e$. With this in mind, we consider linear subspaces
of $\CC^m$ to model rank defects in $A(x)$.

Let $A(x) \in \ss_m^{n+1}(\QQ)$, and let $p,r \in \NN$, with $0 \leq p \leq
r \leq m$. We denote by $Y(y)=(y_{i,j})$ a $m \times (m-p)$ matrix with unknowns entries. Then,
for $x^* \in \CC^n$, $A(x^*)$ has rank at most $p$, if and only if there is $y^* \in \CC^{m(m-p)}$
such that $A(x^*)Y(y^*)=0$, with $\rank \,Y(y^*)=m-p$. Moreover, one can suppose that one of the maximal minors of $Y(y^*)$
is the identity matrix $\Id_{m-p}$ ({\it cf.} for example \cite[Sec.\,2]{FauSafSpa}).

For $\iota \subset \{1, \ldots, m\}$ with $\#\iota = m-p$, we denote by $Y_\iota$ the
maximal minor of $Y(y)$ whose rows are indexed by $\iota$. We deduce that $\calD_p$
is the image under the projection $\pi_n \colon \CC^n \times \CC^{m(m-p)} \to \CC^n$ of
the {algebraic set}
\[
\calV_p=
\bigcup_{{\begin{array}{c} \iota \subset \{1, \ldots, m\} \\ \#\iota=m-p \end{array}}} 
\calV_{p,\iota}
\]
where $\calV_{p,\iota}=\{(x,y) \in \CC^n \times \CC^{m(m-p)} \mymid A(x)Y(y)=0,
Y_\iota=\Id_{m-p}\}$. We call the sets $\calV_{p,\iota}$ {\it incidence varieties} for $\calD_p$.
We denote by $f(A,\iota)$ (often simply by $f$) the polynomial
system defining $\calV_{p,\iota}$. We prove the following Proposition on the
regularity of $\calV_{p,\iota}$.

\begin{proposition}
\label{smoothness}
Let $\iota \subset \{1,\ldots,m\}$ with $\#\iota = m-p$.
\begin{enumerate}
\item There is a subsystem $f_{red} \subset f(A,\iota)$ of cardinality $\# f_{red} = m(m-p)+\binom{m-p+1}{2}$
  such that $\zeroset{f_{red}}=\zeroset{f(A,\iota)}=\calV_{p,\iota}$.
\item There is a non-empty Zariski open set $\zarA \subset \ss_m^{n+1}(\CC)$ such that,
  if $A \in \zarA \cap \ss_m^{n+1}(\QQ)$, $\calV_{p,\iota}$ is either empty or smooth and equidimensional of co-dimension
  $m(m-p)+\binom{m-p+1}{2}$, and $f$ generates a radical ideal.
\end{enumerate}
\end{proposition}
\begin{proof}
We start with Point 1, by explicitely constructing the subsystem $f_{red}$.
Suppose w.l.o.g. that $\iota = \{1, \ldots, m-p\}$, and denote by $g_{i,j}$
the $(i,j)-$th entry of the matrix $A(x)Y(y)$ where $Y_\iota$ has been
substituted by $\Id_{m-p}$. Then $f_{red}$ is defined as follows:
$f_{red} = (g_{i,j} \,\, \text{for} \,\, i \geq j, Y_\iota-\Id_{m-p})$.

We prove now that $\zeroset{f_{red}}=\zeroset{f(A,\iota)}$. If $a_{i,j}$ is
the $(i,j)-$th entry of $A$, for $i<j$ one has that
$g_{i,j}-g_{j,i} = \sum_{\ell=m-p+1}^{m}a_{i,\ell}y_{\ell,j}-a_{j,\ell}y_{\ell,i}$,
since $A$ is symmetric. Using the polynomial relations $g_{k,\ell}=0$ for $k>m-p$ one
can solve for $a_{i,\ell}$ and $a_{j,\ell}$, and deduce
\[
\begin{aligned}
& g_{i,j} - g_{j,i} \equiv \\
& \equiv \sum_{\ell=m-p+1}^m \left(-\sum_{t=m-p+1}^{m}a_{\ell,t}y_{t,i}y_{\ell,j}+\sum_{t=m-p+1}^{m}a_{\ell,t}y_{t,j}y_{\ell,i} \right)\\
 & \equiv \sum_{\ell,t=m-p+1}^m a_{\ell,t} \left( -y_{t,i}y_{\ell,j}+y_{t,j}y_{\ell,i} \right) \equiv 0
\end{aligned}
\]
modulo $\left\langle g_{k,\ell}, \, k>m-p \right\rangle$. This proves Point 1.

We now give the proof of Point 2.
We denote by $\varphi$ the polynomial map $\colon \CC^{n+m(m-p)} \times \ss_m^{n+1}(\CC) \to \CC^{m(m-p)+\binom{m-p+1}{2}}$
sending $(x,y,A)$ to $f_{red}(x,y,A)$, and let $\varphi_A$ denote the section map $\varphi_A(x,y)=\varphi(x,y,A)$.
Hence $\varphi_A^{-1}(0)=\calV_{p,\iota}$. If $\varphi^{-1}(0)=\emptyset$, then for all $A \in
\ss^{n+1}_m(\CC)$, $\varphi_A^{-1}(0)=\calV_{p,\iota}=\emptyset$, and we conclude defining $\zarA=\ss_m^{n+1}(\CC)$.

If $\varphi^{-1}(0)\neq\emptyset$, we prove below that $0$ is a regular value of $\varphi$. We
deduce by Thom's Weak Transversality Theorem \cite[Sec.4.2]{din2013nearly} that there exists a non-empty
Zariski open set $\zarA_\iota \subset \ss_m^{n+1}(\CC)$ such that for $A \in \zarA_\iota$, $0$ is
a regular value of $\varphi_A$. We finally deduce by the Jacobian Criterion
\cite[Th.16.19]{Eisenbud95} that for $A \in \zarA_\iota$, $\calV_{p,\iota}$ is smooth and equidimensional
of co-dimension $m(m-p)+\binom{m-p+1}{2}$, and that the ideal generated by $f_{red}$ is radical.
We conclude defining $\zarA = \cap_{\iota} \zarA_\iota$.

Now we only have to prove that $0$ is a regular value of $\varphi$.
Let $\jac \varphi$ be the Jacobian matrix of $\varphi$. We denote by $a_{\ell,i,j}$ the
variable representing the $(i,j)-$th entry of $A$. We consider the derivatives of elements in $f_{red}$
with respect to:
\begin{itemize}
\item
  the variables $\eta=\{a_{0,i,j} \mymid i \leq m-p \,\,\, \text{or} \,\,\, j \leq m-p\}$;
\item
  the variables $y_{i,j}$ with $i \in \iota$.
\end{itemize}
Let $(x,y,A) \in \varphi^{-1}(0)$. The submatrix of $\jac \varphi(x,y,A)$ containing such derivatives,
contains the following non-singular blocks: the derivatives of $A(x)Y(y)$ w.r.t. elements in $\eta$,
that is a unit block $\Id_{(m-p)(m+p+1)/2}$; the derivatives of $Y_\iota-\Id_{m-p}$, that is a unit block
$\Id_{(m-r)^2}$. These two blocks are orthogonal, and we deduce that $\jac\varphi$ is full rank at
the point $(x,y,A)$. Since $(x,y,A)$ is arbitrary in $\varphi^{-1}(0)$, we conclude that $0$ is a regular
value of $\varphi$.
\end{proof}

\begin{example}
\label{exampleredundancies}
We construct an example of the relations among the polynomials defining $\calV_{p,\iota}$, computed by
Let $A(x) = (x_{i,j})_{i,j}$ be a $3 \times 3$ symmetric matrix of unknowns $x=(x_{11},x_{12},x_{13},x_{22},x_{23},x_{33})$.
We encode matrices of rank 1 in the pencil $A(x)$ with kernel configuration $\iota = \{1,2\} \subset \{1,2,3\}$
via the following polynomial equations:
\[
%\left(
%\begin{array}{ccc}
%f_{11} & f_{12} \\
%f_{21} & f_{22} \\
%f_{31} & f_{32} 
%\end{array}
%\right)
%=
A(x) \cdot
\left(
\begin{array}{ccc}
1 & 0 \\
0 & 1 \\
y_{31} & y_{32} 
\end{array}
\right) = 0.
\]
Denoting with $f_{ij}$ the $(i,j)-$th entry of the previous matrix product, it is straightforward to check that 
$ f_{12}-f_{21} = y_{32}x_{3}-y_{31}x_{5} \equiv y_{31}x_{6}y_{32}-y_{32}x_{6}y_{31} = 0 $, modulo the ideal
$I=\langle f_{31}, f_{32} \rangle$.
\end{example}

\subsection{Critical points}
\label{ssec:critpoints}
%We will compute the solution of $\sdp_r$ as one of the critical points of the linear function
%$c^Tx$ restricted to $\calD_r$ (see Theorem \ref{reduction} below). To do that, we consider
%the critical points of $c^Tx$ restricted to every component $\calV_{r,\iota}$ of the incidence
%variety $\calV_r$. Since the image of the projection of $\calV_r$ on the $x-$variables equals
%by construction $\calD_r$, and since we are considering projections over lines in $\CC^n$,
%the critical points of $c^Tx$ restricted to $\calV_r$ project to the critical points of $c^Tx$
%restricted to $\calD_r$. This is proved in the next Lemma.
%
%\begin{lemma}
%ss
%\end{lemma}

In this section we consider polynomial systems encoding the local minimizers of the linear function
$\ell_c(x) \colon \RR^n \to \RR$ in \eqref{sdpr} restricted to the determinantal variety $\calD_p \cap \RR^n$,
with $0 \leq p \leq r$.
We denote by $L_c$ the map $L_c \colon \RR^{n+m(m-p)} \to \RR$ sending $(x,y)$ to $c^Tx$, that is
$L_c = \ell_c \circ \pi_n$, with $\pi_n \colon \RR^{n+m(m-p)} \to \RR^n$, $\pi_n(x,y)=x$. With analogy to the
description of $\calD_p$ via incidence varieties of the previous section, we consider the set
$\crit(\ell_c,\calV_{p,\iota} \cap \RR^{n+m(m-p)})$ of critical points of the restriction of $L_c$ to
$\calV_{p,\iota}\cap \RR^{n+m(m-p)}$.

\begin{lemma}
\label{lemmaProjCrit}
Let $\zarA \subset \ss_m^{n+1}(\CC)$ be the Zariski open set given in Proposition \ref{smoothness},
and let $A \in \zarA$. The set of local minimizers of $\ell_c$ on $\calD_p \cap \RR^n$ is contained
in the image of the union of the sets $\crit(L_c,\calV_{p,\iota})$, for $\iota \subset \{1, \ldots,
m\}$, with $\#\iota=m-p$, via the projection map $\pi_n(x,y)=x$.
\end{lemma}
\begin{proof}
Let $\tilde{x} \in \RR^n$ be a local minimizer of $\ell_c$ on $\calD_p \cap \RR^n$, and let
$\cc_{\tilde{x}} \subset \calD_p\cap \RR^n$ be the connected component containing $x$. Let
$t=\ell_c(\tilde{x})$. Then $\ell_c(x) \geq t$ for all $x \in U \cap \cc_{\tilde{x}}$, for
some $U$ connected open set. By definition
of $\calV_p$, and since $\tilde{x} \in \calD_p$, there exists $\iota\subset
\{1, \ldots, m-p\}$ and $\tilde{y} \in \RR^{m(m-p)}$ such that $(\tilde{x},\tilde{y}) \in
\calV_{p,\iota}$. Let $\cc_{(\tilde{x},\tilde{y})}$ be the connected component of $\calV_{p,\iota}
\cap \RR^{n+m(m-p)}$ containing $(\tilde{x},\tilde{y})$. We claim (and prove below) that
$(\tilde{x},\tilde{y})$ is a minimizer of $L_c$ on $\pi_n^{-1}(U)\cap \cc_{(\tilde{x},\tilde{y})}$, hence
local minimizer on $\pi_n^{-1}(U)\cap \calV_{p,\iota}$. We deduce that $t = \ell_c(\tilde{x})=L_c(\tilde{x},
\tilde{y})$ lies in the boundary of $L_c(\pi_n^{-1}(U) \cap \cc_{\tilde{x},\tilde{y}})$. In particular, the differential
map of $L_c$ at $x$ is not surjective: because $A \in \zarA$, then $\calV_{p,\iota}$ is smooth
and equidimensional, and hence $(\tilde{x},\tilde{y}) \in \crit(L_c,\calV_{p,\iota} \cap
\RR^{m(m-p)})$.

Now we prove our claim. Recall that $L_c(\tilde{x},\tilde{y})=\ell_c(\tilde{x})=t$, and suppose that there is
$(x,y) \in \pi_n^{-1}(U) \cap\cc_{(\tilde{x},\tilde{y})}$ such that $L_c(x,y)<t$. There exists
a continuous semialgebraic map $\tau \colon [0,1] \to
\cc_{(\tilde{x},\tilde{y})}$ such that $\tau(0)=(\tilde{x},\tilde{y})$ and $\tau(1)=({x},{y})$.
We deduce that $\pi_n \circ \tau$ is also continuous and semialgebraic. %%%%%%% (since $\pi_n$ is).
Since $\pi_n \circ \tau(0)=\tilde{x}$ and $\pi_n \circ \tau(1)={x}$, one gets $x \in
U \cap \cc_{\tilde{x}}$. Then $\ell_c(x) = L_c(x,y) < t = \ell_c(\tilde{x})$ contradicts the
hypothesis that $\tilde{x}$ is a local minimizer of $\ell_c$ on $\cc_{\tilde{x}}$.
\end{proof}
Lemma \ref{lemmaProjCrit} states that the minimizers of $\ell_c$ on $\calD_p \cap \RR^n$ are
obtained as the projection on the first $n$ variables of the critical points of $L_c$ over
the lifted incidence variety $\calV_p \cap \RR^{n+m(m-p)}$. We are now going to prove that
such critical points are generically finite. Let us suppose that $A \in \zarA$ (see Proposition
\ref{smoothness}), and let $c \in \QQ^n$. We also fix a subset $\iota \subset\{1,\ldots,m\}$
of cardinality $\#\iota=m-p$.

We have denoted, in Section \ref{incvar}, by $f \subset \QQ[x,y]$ the polynomial system
defining $\calV_{p,\iota}$, constituted by the entries of $A(x)Y(y)$ and of $Y_\iota-\Id_{m-p}$.
By Proposition \ref{smoothness}, we deduce that $f_{red}$, and hence $f$, generates a radical
ideal and defines a smooth equidimensional algebraic set of co-dimension $m(m-p)+\binom{m-p+1}{2}$.
The set $\crit(L_c,\calV_{p,\iota})$ is hence defined (after the elimination of the Lagrange multipliers)
by the following polynomial system:
\begin{equation}
\label{LagSys}
\lag(\iota):
\qquad
f = 0;
\,\,\,\,\,\,\,
(g,h) =
z'
\left[
\begin{array}{c}
\jac f \\
\jac L_c
\end{array}
\right]
%=
%z'
%\left[
%\begin{array}{cc}
%\jac_x f & \jac_y f \\
%c' & 0
%\end{array}
%\right] 
= 0,
\end{equation}
where $z = (z_1,\ldots,z_{(2m-p)(m-p)},1)$ is the vector of Lagrange multipliers: these are the classical
first-order optimality conditions in constrained optimization. In the
previous notation, the vector $g$ (resp. $h$) is of size $n$ (resp. $m(m-p)$). For the sake of brevity,
we say that a point $(x,y,z) \in \zeroset{\lag(\iota)}$ has rank $p$, if $\rank A(x) = p$.

Our next goal in this section is to prove the following Proposition. It states that if the linear function
$\ell_c$ in Problem \eqref{sdpr} is generic, the points $x^* \in \calD_p \cap \RR^n$, such that $\rank A(x^*)=p$,
that correspond to critical points $(x^*,y^*)$ of the restriction of $L_c$ to $\calV_p \cap \RR^{n+m(m-p)}$, are
finitely many.

\begin{proposition}
\label{finiteness}
Let $\zarA \subset \ss_m^{n+1}(\CC)$ be the Zariski open set defined by Proposition \ref{smoothness}, and
let $A \in \zarA \cap \ss_m^{n+1}(\QQ)$. There exists a non-empty Zariski open set $\zarC \subset \CC^n$ such
that, for $c \in \zarC \cap \QQ^n$, for every $p = 0, \ldots, r$, and for every $\iota \subset \{1, \ldots, m\}$
such that $\#\iota = m-p$, the projection of the solutions of the system $\lag(\iota)$ of rank $p$ over the
$x-$space is a finite set.
\end{proposition}

In order to prove Proposition \ref{finiteness}, we use the local description of determinantal varieties as
developed in \cite[Sec.\,4.1]{HNS2014} and in \cite[Sec.\,5.1]{HNS2015a}. This is briefly recalled below.
Suppose that $x \in \calD_p \cap \RR^n$, with $\rank\,A(x)=p$, and that the upper-left $p \times p$ submatrix
$N$ of $A(x)$ is non-singular (at least one of the $p \times p$ submatrices of $A(x)$ is non-singular).
That is
\begin{equation}
A(x)
=
\left[
\begin{array}{cc}
N & Q \\
P & R
\end{array}
\right]
\end{equation}
and $\det N \neq 0$. Suppose also w.l.o.g. that $\iota = \{1, \ldots, m-p\}$. By \cite[Sec.4.1]{HNS2014} or
\cite[Lemma 13]{HNS2015b}, the local equations of $\calV_{p,\iota}$ over $x$ are given by
\begin{equation}
\label{localEq1}
\left[
\begin{array}{cc}
\Id_p & N^{-1}Q \\
0 & \Sigma(N)
\end{array}
\right]
Y(y) = 0
\qquad
\text{and}
\qquad
Y_\iota-\Id_{m-p}=0,
\end{equation}
where $\Sigma(N) = R-PN^{-1}Q$ is the Schur complement of $A(x)$ at $N$, well defined since $N$ is not singular:
these are elements of the local ring $\QQ[x,y]_{\det\,N}$ at $I=\langle \det\,N \rangle$.
Let $Y^{(1)}$ (resp. $Y^{(2)}$) be the matrix obtained by isolating the first $p$ rows (resp. last $m-p$ rows) from
$Y(y)$. Let $U_\iota$ be such that $U_\iota Y(y) = Y_\iota$, and let $U_\iota = [U^{(1)}_\iota | U^{(2)}_\iota]$ be the
corresponding column subdivision of $U_\iota$. Then \eqref{localEq1} imply
$
\Id_{m-p} = U^{(1)}_\iota Y^{(1)}+U^{(2)}_\iota Y^{(2)} = (U^{(2)}_\iota -U^{(1)}_\iota N^{-1}Q) Y^{(2)}
$
and hence that both $Y^{(2)}$ and $U^{(2)}_\iota-U^{(1)}_\iota N^{-1}Q$ are invertible (in the local ring $\QQ[x]_{\det\,N}$).
We deduce the following equivalent form of the previous equations:
\begin{eqnarray}
\tilde{f}: \,\,\,\,\,\,\,\,\,\,\,\,\,
&& Y^{(1)}+N^{-1}QY^{(2)} = 0, \,\,\,\,\,\,\,\, \Sigma(N) = 0, \,\, \nonumber \\
&& Y^{(2)}-(U^{(2)}_\iota-U^{(1)}_\iota N^{-1}Q)^{-1} = 0,
\end{eqnarray}
denoted by $\tilde{f}$. Up to reordering its entries, the Jacobian matrix of $\tilde{f}$ is
\[
\jac \tilde{f} =
\left[
\begin{array}{cc}
D_x[\Sigma(N)]_{i,j} & 0_{(m-p)^2 \times m(m-p)} \\
\star & 
\begin{array}{cc}
\Id_{p(m-p)} & \star \\
0 & \Id_{(m-p)^2}
\end{array}
\end{array}
\right].
\]
If $A \in \zarA$, by Proposition \ref{smoothness} the rank of $\jac \tilde{f}$ equals $\# f_{red} = m(m-r)+\binom{m-r+1}{2}$
at every $x \in \zeroset{\tilde{f}}$. Similarly, we localize the Lagrange system $\lag(\iota)$ ({\it cf.} \eqref{LagSys})
by defining:
\[
(\tilde{g},\tilde{h}) = z'
\left[
\begin{array}{c}
\jac \tilde{f} \\
\jac L_c
\end{array}
\right].
\]
By the structure of $\jac \tilde{f}$, one gets $\tilde{h}_i=z_{(m-p)^2+i}$, for $i=1, \ldots, m(m-p)$,
and hence one can substitute $z_{(m-p)^2+i}=0, i=1, \ldots, m(m-p)$, in $(\tilde{f},\tilde{g})$.

\begin{proof}[Proof of Proposition \ref{finiteness}]
Let $d=m(m-p)+\binom{m-p+1}{2}$ and $e=\binom{m-p}{2}$ so that $d+e=(2m-p)(m-p)=\# z$.
First, we claim that there exists a non-empty Zariski open set $\zarC_N \subset \CC^n$ such that
if $c \in \zarC_N \cap \QQ^n$ the Jacobian matrix of the local system $(\tilde{f},\tilde{g},\tilde{h})$ has
maximal possible rank. Here $N$ refers to the upper left $p \times p$ submatrix of $A$ as above.
We conclude by defining $\zarC = \cap_N \zarC_N$ (where $N$ runs over the family of $p \times p$ submatrices
of $A$), which is non-empty and Zariski open.

The proof is similar to that of Point 2 of Proposition \ref{smoothness} and hence we only sketch it. Let
\[
  \begin{array}{lrcc}
  \varphi : &  \CC^{n+d+e+m(m-p)} \times \CC^{n} & \longrightarrow & \CC^{n+d+e+m(m-p)} \\
            &  (x,y,z,c) & \longmapsto & (\tilde{f}, \tilde{g}, \tilde{h})(x,y,z,c).
  \end{array}
\]
Then the Jacobian matrix of $(\tilde{f}, \tilde{g}, \tilde{h})$ is $\jac \varphi$
as a polynomial map.
We prove that $0$ is a regular value of $\varphi$, and apply Thom's Weak Transversality
Theorem \cite[Sec.4.2]{din2013nearly} as in the proof of Proposition \ref{smoothness}.
Let $(x,y,z,c) \in \varphi^{-1}(0)$ (if it does not exist, define $\zarC_N=\CC^n$).
Since polynomials in $\tilde{f}$ only depend on $x$ and $y$, then $\jac \tilde{f}$ is a
submatrix of $\jac \varphi$ and the columns corresponding to the derivatives of $\tilde{f}$
with respect to $z,c$ are zero. Hence the rank of $\jac \varphi$ is at most $n+d+m(m-r)$
since $\jac \tilde{f}$ has $e$ rank defects by Proposition \ref{smoothness} (recall that
$A \in \zarA$). A
%%%%%%%%%%%%%%$(n+d+m(m-p)) \times (n+d+m(m-p))$
full-rank submatrix of $\jac\varphi$ at $(x,y,z,c)$ is then given in this case by the
derivatives with respect to: (1) $x,y$, (2) $c_1, \ldots, c_n$, and (3) $z_{(m-p)^2+i}, i=1,
\ldots, m(m-p)$. 

Now, we can conclude the proof. Let $c \in \zarC = \cap_N \zarC_N$ (previously defined).
From the previous claim, we deduce that the locally closed set ${\cal E} = \zeroset{\lag(\iota)}
\cap \{(x,y,z) \mymid \rank\,A(x)=p\}$ is empty or equidimensional of dimension $e$.
Let
\[
\begin{array}{lrcc}
  \pi : &  \CC^{n+m(m-p)+d+e} & \longrightarrow & \CC^{n} \\
            &  (x,y,z) & \longmapsto & x
\end{array}
\]
be the projection over the $x-$space, and $x^* \in \pi(\cal E)$.
In particular $\rank A(x^*) = p$, and there is a unique $y^* \in \CC^{m(m-p)}$
such that $f(x^*,y^*) = 0$. We deduce that $\pi^{-1}(x^*)$ is isomorphic to the linear space defined by
\[
\Big\{(z_1, \ldots, z_{d+e}) \mymid (z_1, \ldots, z_{d+e}) \jac f = (c', 0)\Big\}.
\]
Since the rank of $\jac f$ is $d$, $\pi^{-1}(x^*)$ is a linear space of dimension
$e$, and by the Theorem on the Dimension of Fibers \cite[Sect. 6.3, Theorem 7]{Shafarevich77}
$\pi_x(\mathcal{E})$ has dimension $0$.
\end{proof}

\section{From se\-mi-\-al\-geb\-raic to al\-geb\-raic op\-ti\-mi\-za\-tion}
\label{sec-fromsemialg}
In order to prove that our algorithm is correct, we present in this section the main geometric
result of this work.
By the independent interest of the results of this section, we need to introduce, first, some
notation.

Given $c \in \QQ^n$ and $A \in \ss_m^{n+1}(\QQ)$, for $0 \leq r \leq m$, we have denoted by $\calF_r(A,c)$
the (possibly empty or infinite) set of minimizers of $\ell_c$ on $\spec \cap \calD_r$. By simplicity,
we also call $\calF_r(A,c)$ the set of minimizers of $\sdp_r$. When $r=m$, $\calF_m(A,c)$ is
the convex optimal face of the spectrahedron $\spec$ in direction $c$. Indeed, since every face of
a spectrahedron is exposed, it is exactly defined as the set of minimizers of some semidefinite program
$\sdp_m$. We denote by
\[
\calR_r(A,c) = \bigg\{p \mymid 0 \leq p \leq r, \, \exists\,x\in\calF_r(A,c), \, \rank\,A(x)=p \bigg\}
\]
the rank profile of $\calF_r(A,c)$, namely the set of ranks of matrices in $\calF_r(A,c)$. Clearly,
$\calF_r(A,c) \neq \emptyset$ if and only if $\calR_r(A,c) \neq \emptyset$. This is our main theorem in
this section.

\begin{theorem}
\label{reduction}
Suppose that $\calF_r(A,c) \neq \emptyset$, and let $p \in \calR_r(A,c)$. For $x^* \in \calF_r(A,c)$
such that $\rank\,A(x^*)=p$, then $x^*$ is a local minimizer of $\ell_c$ on $\calD_{p} \cap \RR^n$.
\end{theorem}

\begin{proof}%%[of Theorem \ref{reduction}]
Suppose that $x^*$ is as in the hypothesis. We denote by $\cc^* \subset \calD_p \cap \RR^n$
the connected component of $\calD_p \cap \RR^n$ containing $x^*$. Hence there are three possible
(non mutually exclusive) cases, that we analyze below. Recall that $p \leq r$, hence
$\calD_p \subset \calD_r$.

{\it First case}: $\cc^* \subset \spec$. Hence $\cc^* \subset \spec \cap \calD_p \subset \spec
\cap \calD_r$. Since $\spec \cap \calD_r$ is the feasible set of $\sdp_r$ and $x^*$ is a minimizer
of $\sdp_r$, hence $x^*$ is a minimizer of $\ell_c$ on $\cc^*$. Hence it is a local minimizer of
$\ell_c$ on $\calD_p \cap \RR^n$, as claimed.

{\it Second case}: There exists an open set $U \subset \RR^n$ such that $x^* \in U$ and
$U \cap (\calD_{m-1} \setminus \spec) = \emptyset$. This means that $U$ intersects
$\calD_{m-1} \cap \RR^n$ only at positive semidefinite matrices, and $U \cap \spec$ is an
open subset of  $\spec$ containing $x^*$. We deduce that $x^*$ is a 
minimizer of $\ell_c$ on $U \cap \calD_{p} \subset U \cap \spec$, hence a local minimizer of
$\ell_c$ on $\calD_p \cap \RR^n$.

{\it Third case}: $\cc^* \not\subset \spec$, and for all $U \subset \RR^n$ open set, such that
$x^* \in U$, then $U \cap (\calD_{m-1} \setminus \spec) \neq \emptyset$. We prove below
that such a situation cannot occur. Indeed, one first deduces that, for all $U$ as above,
$U \cap (\calD_{p} \setminus \spec) \neq \emptyset$ since $\cc^* \not\subset \spec$.
For a positive integer $d \in \NN$, we denote by $B(x^*,1/d)$ the open ball with center
$x^*$ and radius $1/d$, that is $B(x^*,1/d)=\{x \in \RR^n \mymid \|x-x^*\|<1/d\}$, where
$\|x\|$ is the Euclidean norm of $x$. By hypothesis, for all $d \in \NN$ there exists
$x(d) \in B(x^*,1/d) \cap \calD_p$ such that $A(x(d)) \not \succeq 0$. Hence $x(d)
\rightarrow x^*$ when $d \rightarrow \infty$. Denoting by $e_1(x) \leq e_2(x) \leq \cdots
\leq e_m(x)$ the ordered eigenvalues of $A(x)$, one deduces that, for all $d \in \NN$,
$e_1(x(d))<0$ and hence $e_{m-p+1}(x(d)) \leq 0$ (since the matrix $A(x(d))$ has at least
$m-p$ null eigenvalues). In particular $e_{m-p+1}(x(d)) \rightarrow e_{m-p+1}(x^*) \leq 0$
when $d \rightarrow \infty$. Since $x^* \in \spec$, then $e_1(x^*)=\cdots=e_{m-p}(x^*)=
e_{m-p+1}(x^*)=0$, and the rank of $A(x^*)$ is at most $p-1$, which contradicts the
hypotheses.
\end{proof}

We prove two corollaries of Theorem \ref{reduction} and of previous results, which are
worth to be made explicit and highlighted.

\begin{corollary}
Let $x^* \in \calF_r(A,c)$ satisfy the following property: for all Euclidean open sets
$U \subset \RR^n$ containing $x^*$, $U$ contains a singular matrix with a negative eigenvalue. Then,
if $p = \rank\,A(x^*)$, the connected component $\cc^* \subset \calD_p \cap \RR^n$ containing $x^*$
is contained in $\spec$.
\end{corollary}
\begin{proof}
We apply {\it mutatis mutandis} the argument of the Third case in the proof of Theorem
\ref{reduction}, without the hypothesis that $\cc^* \not\subset \spec$. Hence we conclude
that necessarily $\cc^* \subset \spec$.
\end{proof}
% put that after Lagrange systems
%The second, is a sufficient condition for Assumption \ref{ass-finite} to hold.
%We first need to give the definition of critical points. 
%
%\begin{proposition}
%\end{proposition}
%\begin{proof}
%\end{proof}

The second corollary gives a finiteness theorem for the set of solutions of a generic
rank constrained semidefinite program \eqref{sdpr}.

\begin{corollary}
\label{finitesolsdpr}
Let $\zarA \subset \ss_m^{n+1}(\CC)$ and $\zarC \subset \CC^n$ be the Zariski open sets defined respectively in
Proposition \ref{smoothness} and \ref{finiteness}. If $A \in \zarA \cap \ss_m^{n+1}(\QQ)$ and $c \in \zarC \cap \QQ^n$, the
set $\calF_r(A,c)$ of minimizers of the rank-constrained semidefinite program $\sdp_r$ is finite.
\end{corollary}
\begin{proof}
Remark that $\calF_r(A,c)$ is the union of sets $B_p \subset \calF_r(A,c)$, for $p \in \calR_r(A,c)$,
corresponding to minimizers of rank $p$, that is $\calF_r(A,c)=\cup_{p \in \calR_r(A,c)} B_p$. We prove
that $B_p$ is finite for all $p \in \calR_r(A,c)$.

Let $x^* \in B_p$. By Theorem \ref{reduction}, $x^*$ is a local minimizer of $\ell_c$ on $\calD_p \cap \RR^n$.
Since $A \in \zarA$, by Lemma \ref{lemmaProjCrit} $B_p$ is included in the union of the projections of
the sets of critical points of $L_c$ on $\calV_{p,\iota}$, for $\iota \subset \{1, \ldots, m\}, \#\iota=m-p$.
Since $c \in \zarC$, and since $\rank\,A(x^*)=p$, by Proposition \ref{finiteness} $B_p$ is the projection
of a finite set, hence finite.
\end{proof}

\section{The algorithm}
\label{algorithm}
The main algorithm described in this work is called {\sc SolveSDP}.
\subsection{Description}
\label{formalDescr}
We first describe the main subroutines of {\sc SolveSDP}.

\smallskip
\noindent \,\, {\sf CheckReg}. With input $A \in \ss_{m}^{n+1}(\QQ)$ and $p \leq r$, it returns {\tt true}
  if for all $\iota \subset \{1, \ldots, m\}$, with $\#\iota=m-p$, the set $\calV_{p,\iota}$
  is smooth and equidimensional; otherwise, it returns {\tt false}.

\smallskip
\noindent \,\, {\sf Optimize}. With input $A,c$ and $p$, it returns the vector of ideals $(\lag(\iota_1),
\ldots,\lag(\iota_{\binom{m}{p}})) \subset \QQ[x,y,z]$, where $\iota_j \subset \{1,\ldots,m\}$, with $\#\iota_j=m-p$,
$j=1, \ldots, \binom{m}{p}$. The set $\cup_j\zeroset{\lag(\iota_j)}$ encodes the union of the critical
points of $L_c$ restricted to the components $\calV_{p,\iota}$ of $\calV_p$.

\smallskip
\noindent \,\, {\sf Project}. With input the output of {\sf Optimize}, it substitutes each ideal $\lag(\iota_j)$
with the elimination ideal $I_{\iota_j}=\lag(\iota_j) \cap \QQ[x]$, for $j=1, \ldots, \binom{m}{p}$, returning
$I=(I_{\iota_j}, i=1, \ldots, \binom{m}{p})$.

We recall the definition of rational parametrization of a finite set $S \subset \RR^n$: this is given by a vector
$Q=(q,q_0,q_1, \ldots, q_n) \subset \QQ[t]$ such that $S$ admits a representation \eqref{ratrep}.
We need to define two routines performing operations on rational parametrizations of finite sets.

\noindent \,\, {\sf RatPar}. Given a zero-dimensional ideal $I_{\iota_j} \subset \QQ[x]$, it returns a rational
parametrization $Q = (q,q_0,q_1, \ldots, q_n) \subset \QQ[t]$ of $I_{\iota_j}$. If $I_{\iota_j}$ is not zero-dimensional,
it returns an error message.

\smallskip
\noindent \,\, {\sf Union}. Given rational parametrizations $Q_1,Q_2 \subset \QQ[t]$ encoding two finite sets
$V_1,V_2 \subset \CC^n$, it returns a rational parametrization $Q \subset \QQ[t]$ encoding $V_1 \cup V_2$.

\smallskip
The following is the formal procedure of {\sc SolveSDP}. We offer below a more explicit description of the
algorithm for the sake of clarity.

\begin{algorithm}
\caption{\bf SolveSDP}
\label{solvesdp}
\begin{algorithmic}[1]
\Procedure{SolveSDP}{$A,c,r$}
\State $Q \leftarrow [\,\,]$
\For{$p=0,\ldots,r$}
\If {${\sf CheckReg}(A,p)=$ {\tt false}} \Return{error}
\EndIf
\State $I \leftarrow {\sf Project}({\sf Optimize}(A,c,p))$
\For{$j=1, \ldots, \binom{m}{p}$}
\State $Q_{\iota_j} \leftarrow {\sf RatPar}(I_{\iota_j})$
\State $Q \leftarrow {\sf Union}(Q,Q_{\iota_j})$
\EndFor
\EndFor
\State \Return{$Q$}
\EndProcedure
\end{algorithmic}
\end{algorithm}

The input is a triple $(A,c,r)$, where $A \in \ss_m^{n+1}(\QQ)$ is $(n+1)-$tuple of symmetric matrices with
rational coefficients, $c \in \QQ^n$ defines the linear function $\ell_c$ in \eqref{sdpr} and $r$ is
the maximum admissible rank. For every value of $p$ from $0$ to $r$, the algorithm checks whether
the regularity assumption on the incidence varieties $\calV_{p,\iota}, \iota\subset\{1, \ldots, m\},$ for
$\#\iota=m-p$, holds. If this is the case, it computes rational pa\-ra\-me\-tri\-za\-tions $Q_{\iota}$ of the Lagrange systems
encoding the critical points
%%%%%%%of rank $p$
of the map $L_c$, on the components $\calV_{p,\iota}$ of the incidence variety $\calV_p$. The output
is a rational parametrization $Q$ encoding the union of the finite sets defined by the $Q_{\iota}'$s.

\subsection{Correctness}

We prove in this section that {\sc SolveSDP} is correct. Our proof relies on intermediate results already
stated and proved in the previous sections.
%%We recall that in Proposition \ref{smoothness} we have defined the Zariski open
%%set $\zarA \subset \ss_m^{n+1}(\CC)$ and ...

\begin{theorem}
\label{correctness}
Let $\zarA \subset \ss_m^{n+1}(\CC)$ and $\zarC \subset \CC^n$ be the Zariski open sets defined respectively
by Proposition \ref{smoothness} and \ref{finiteness}. Let $A \in \zarA \cap \ss_m^{n+1}(\QQ)$, $c \in \zarC
\cap \QQ^n$ and $0 \leq r \leq m$. Then the output of {\sc SolveSDP} is a rational parametrization of a
finite set containing all minimizers of $\sdp_r$.
\end{theorem}
\begin{proof}
Let $(A,c,r)$ be the input of {\sc SolveSDP}, and let $x^* \in \RR^n$ be a solution of $\sdp_r$.
Let $p=\rank A(x^*)$. By Theorem \ref{reduction}, $x^*$ is a local minimizer of $\ell_c$ on
$\calD_p \cap \RR^n$. Let us denote by $S$ the image of the union of sets $\crit(L_c,\calV_{p,\iota}),
\iota \subset \{1, \ldots, m\}, \#\iota=m-p$ under the projection $\pi_n(x,y)=x$, namely
\[
S=\pi_n\left(\bigcup_{\#\iota=m-p} \crit(L_c,\calV_{p,\iota}) \right).
\]  
Lemma \ref{lemmaProjCrit} implies that $x^* \in S$. Since $A \in \zarA$, by Proposition \ref{smoothness}
$\calV_{p,\iota}$ is smooth and equidimensional of dimension $m(m-p)+\binom{m-p+1}{2}$. Hence, for all
$\iota \subset \{1, \ldots, m\}$, with $\#\iota=m-p$, the set $\crit(L_c,\calV_{p,\iota} \cap \RR^{n+m(m-p)})$
is defined by the Lagrange system $\lag(\iota)$ introduced in \eqref{LagSys}. We conclude that there
exists $\iota$ as above, and $y^* \in \CC^{n+m(m-p)}$ and $z^* \in \CC^{(2m-p)(m-p)+1}$ such that $(x^*,y^*,z^*)$
is a solution of $\lag(\iota)$ of rank $p$ (indeed, by hypothesis $\rank A(x^*)=p$). By Proposition
\ref{finiteness}, the solutions of rank $p$ of $\lag(\iota)$ are finitely many.

Hence, respectively, the subroutines {\sf Optimize, Project} and {\sf RatPar} compute a rational
pa\-ra\-me\-tri\-za\-tion $Q_{\iota} = (q^{(\iota)},q_0^{(\iota)},\ldots,q_n^{(\iota)}) \subset \QQ[t]$ such that there
exists $t^* \in \RR$ such that
\[
x^*=(q_1^{(\iota)}(t^*)/q_0^{(\iota)}(t^*),\ldots,q_n^{(\iota)}(t^*)/q_0^{(\iota)}(t^*)).
\]
Then the output $Q$ is a rational parametrization containing $x^*$. By the genericity of $x^*$ among the
solutions of $\sdp_r$, we conclude.
\end{proof}

\section{Complexity analysis}
\label{complexityanalysissection}

\subsection{Degree bounds for the output representation}
\label{degboundsoutputrep}

The output of {\sc SolveSDP} is a rational univariate parametrization $Q = (q,q_0,q_1,\ldots,q_n) \subset \QQ[t]$.
For practical purposes, often it is useful to compute an approximation of the coordinates of the minimizers of
Problem \eqref{sdpr}. This can be done by performing real root isolation on the univariate polynomial $q$. Hence
we are interested in bounding the degree of $q$, which is done by the following Proposition.

\begin{proposition}
\label{degbound}
Let $Q = (q, q_0, q_1, \ldots, q_n) \subset \QQ[t]$ be the rational parametrization returned by {\sc SolveSDP}.
Then
\[
\deg\,q \leq \sum_{p=0}^r \binom{m}{p} \theta(m,n,p),
\]
where
\[
\theta(m,n,p) = \sum_{k} \binom{c_p}{n-k} \binom{n-1}{k+c_p-1-p(m-p)} \binom{p(m-p)}{k},
\]
with $c_p=(m-p)(m+p+1)/2$.
\end{proposition}
\begin{proof}
We first prove that $\theta$ gives a bound on the
degree of the ideal generated by $\lag(\iota)$, that is on the degree of the partial rational parametrization $Q_\iota$.
Since $Q$ encodes the union of all algebraic sets defined by the $Q_{\iota}'$s, and since the previous degree does not
depend on $\iota$, we conclude by adding all such bounds (each one multiplied by $\binom{m}{p}$, the number of subset
$\iota$ of cardinality $m-p$). This relies on an equivalent construction of $\lag(\iota)$ which is given below.

Given $p \in \{0, \ldots, r\}$, we fix a subset $\iota \subset \{1, \ldots, m\}$ with $\#\iota=m-p$.
We exploit the multilinearity of the polynomial system $f$ defining the incidence variety $\calV_{p,\iota}$.
First, we eliminate variables $y_{i,j}$, with $i \in \iota$, by substituting $Y_\iota = \Id_{m-p}$;
we also eliminate polynomials $Y_\iota - \Id_{m-p}$ in $f_{red}$ ({\it cf.} Proposition \ref{smoothness}).
One obtains a polynomial system $\widetilde{f}$ of cardinality $c_p:=(m-p)(m+p+1)/{2}$. Moreover, by
construction, $\widetilde{f}$ is constituted by $c_p$ polynomials of bi-degree at most $(1,1)$ with
respect to the groups of variables $x$ and
\begin{equation}
\label{overy}
\overline{y}:=(y_{i,j} \mymid i \notin \iota).
\end{equation}
We also suppose without loss of generality that the linear map $\ell_c$ in Problem \eqref{sdpr}
defines the projection over $x_1$, that is that $c=(1,0,\ldots,0)$. Hence, the system $\lag(\iota)$ is
equivalent to the following. We consider the $c_p$ elements in $\widetilde{f}$. Let $\jac \widetilde{f}$ be the
Jacobian matrix of $\widetilde{f}$ w.r.t. variables $x,\overline{y}$, and let $\jac_1$ be the matrix obtained
by eliminating the first column from $\jac \widetilde{f}$. The critical points of the projection over $x_1$
restricted to $\zeroset{\widetilde{f}}$  are then defined by $\widetilde{f}=0$ and by $z'\jac_1=0$, where
\begin{equation}
\label{overz}
\overline{z}:=(z_1, \ldots, z_{c_p-1}, 1)
\end{equation}
is a non-zero vector of $c_p-1$ Lagrange multipliers.

Hence $\lag(\iota)$ is equivalent to a polynomial system of
\begin{itemize}
\item
$c_p$ equations of bi-degree at most $(1,1,0)$ w.r.t. $x,\overline{y},\overline{z}$; 
\item
$n-1$ equations of bi-degree at most $(0,1,1)$ w.r.t. $x,\overline{y},\overline{z}$; 
\item
$p(m-p)$ equations of bi-degree at most $(1,0,1)$ w.r.t. $x,\overline{y},\overline{z}$.
\end{itemize}
We call this new polynomial system $\widetilde{\lag(\iota)}$. By the Multilinear B\'ezout Theorem ({\it cf.}
for example \cite[Prop.\,11.1.1]{din2013nearly}) the degree of $\widetilde{\lag(\iota)}$ is bounded above by the coefficient
of $s_x^ns_y^{p(m-p)}s_z^{c_p-1}$ in
\[
(s_x+s_y)^{c_p}(s_y+s_z)^{n-1}(s_x+s_z)^{p(m-p)},
\]
which is exactly $\theta(m,n,p)$.
\end{proof}

\subsection{Bounds on the arithmetic complexity}
\label{complexity-section}

Our goal in this section is to bound the number of arithmetic operations over $\QQ$ performed by the main
subroutine of {\sc SolveSDP}, which is the computation of the rational parametrization $Q_\iota$ done by
{\sf RatPar}.
Before that, we give bounds for the complexity of routines {\sf Project} and {\sf Union}. Let
$\widetilde{\lag(\iota)} \subset \QQ[x,\overline{y},\overline{z}]$ ({\it cf.} \eqref{overy} and \eqref{overz})
be the equivalent Lagrange system built in the proof of Proposition \ref{degbound},
and $\theta=\theta(m,n,p)$ be the bound on the degree of $\widetilde{\lag(\iota)}$. From \cite[Chapter\,10]{din2013nearly},
one gets the following estimates:
\begin{itemize}
\item
  by \cite[Lemma\,10.1.5]{din2013nearly}, {\sf Project} can be performed with at most $n^2 \theta(m,n,p)^2$ arithmetic
  operations;
\item
  by \cite[Lemma\,10.1.3]{din2013nearly}, {\sf Union} can be performed with at most $n (\sum_{s=0}^p \binom{m}{s}\theta(m,n,s))^2$
  arithmetic operations.
\end{itemize}

We now turn to the complexity of {\sf RatPar}. Our complexity model is the symbolic homotopy algorithm for
computing rational parametrization in \cite{SymbHom}. This is a probabilistic exact algorithm for solving
zero-dimensional systems via rational parametrizations, exploiting their sparsity. It allows to express
the arithmetic complexity of {\sf RatPar} as a function of geometric invariants of the system $\widetilde{\lag(\iota)}$
(mainly of its degree, which is bounded by $\theta(m,n,p)$, {\it cf.} Proposition \ref{degbound}).

We briefly recall the construction of the homotopy curve in \cite{SymbHom}. This is similar to \cite[Sec.4]{HNS2015a}.
Let $t$ be a new variable, and recall that $\widetilde{\lag(\iota)}$ contains quadratic polynomials with
bilinear structure with respect to the three groups of variables $x,\overline{y},\overline{z}$. Let $g \subset
\QQ[x,\overline{y},\overline{z}]$ be a new polynomial system such that: (1) $\# g = \#\widetilde{\lag(\iota)}$,
(2) the $i-$th entry of $g$ is a polynomial with the same monomial structure as the $i-$th entry of $\widetilde{\lag(\iota)}$,
and (3) the solutions of $g$ are finitely many and known. Since $\widetilde{\lag(\iota)}$ is bilinear in $x,\overline{y},
\overline{z}$, the system $g$ can be obtained by considering suitable products of linear forms in, respectively, $x$,
$\overline{y}$ and $\overline{z}$. The algorithm in \cite{SymbHom} builds the homotopy curve $\zeroset{h}$ defined by
\[
h = t \widetilde{\lag(\iota)} + (1-t) g \subset \QQ[x,\overline{y},\overline{z},t].
\]

The proof of the following lemma is technical and we omit it.
\begin{lemma}
\label{degcurve}
Let $\theta(m,n,p)$ be the bound on the degree of $\zeroset{\widetilde{\lag(\iota)}}$ computed in Proposition \ref{degbound}.
The degree of the homotopy curve $\zeroset{h}$ is in
\[
\bigO((n+c_p+p(m-p)) \min\{n,c_p\} \theta(m,n,p)).
\]
\end{lemma}

The degree of $\zeroset{\widetilde{\lag(\iota)}}$ and of the homotopy curve $\zeroset{h}$ are the main ingredients
of the complexity bound for the algorithm \cite{SymbHom}, which is given by \cite[Prop.\,6.1]{SymbHom}. We use this complexity bound
in our estimate. Indeed, let us denote by
\begin{align*}
\Delta_{xy} & = \{ 1, x_i, y_j, x_iy_j \mymid i=1,\ldots, n, j=1, \ldots, p(m-p) \} \\
\Delta_{yz} & = \{ 1, y_j, z_k, y_jz_k \mymid j=1,\ldots, p(m-p), k=1 \ldots, c_p-1 \} \\
\Delta_{xz} & = \{ 1, x_i, z_k, x_iz_k \mymid i=1,\ldots, n, k=1, \ldots, c_p-1 \}
\end{align*}
the supports of polynomials in $\widetilde{\lag(\iota)}$. To state our complexity result for {\sc SolveSDP},
we suppose that all the regularity assumptions on $A(x)$ are satisfied.
%%%%We avoid in particular to consider the control subroutine {\sf CheckReg} in our complexity analysis.

\begin{theorem}
\label{complexity}
Suppose that $A \in \zarA$ (defined in Proposition \ref{smoothness}). Then {\sc SolveSDP} runs within
\[
\bigO \left( \sum_{p=0}^r \binom{m}{p} (n p c_p (m-p))^5 \theta(m,n,p)^2 \right)
\]
arithmetic operations over $\QQ$, where $c_p=(m-p)(m+p+1)/2$.
\end{theorem}
\begin{proof}
Complexity bounds for subroutines {\sf Project} and {\sf Union} have been computed earlier in Section \ref{complexity-section}.

By \cite[Prop.6.1]{SymbHom}, one can compute a rational pa\-ra\-me\-tri\-za\-tion of $\widetilde{\lag(\iota)}$ within
$\bigO((\tilde{n}^2 N \log \Delta + \tilde{n}^{\omega+1})ee')$ where: $\tilde{n} = n+p(m-p)+c_p-1$ is the number of variables
in $\widetilde{\lag(\iota)}$; $N=c_p \#\Delta_{xy}+(n-1)\#\Delta_{yz}+ p(m-p)\#\Delta_{xz} \in \bigO(npc_p(m-p))$; $\Delta=
\max\{\|q\| \mymid q \in \Delta_{xy} \cup \Delta_{yz} \cup \Delta_{xz}\} \leq \tilde{n}$; finally $e$ is the degree of
$\zeroset{\widetilde{\lag(\iota)}}$ and $e'$ the degree of $\zeroset{h}$, and $\omega$ is the exponent of matrix
multiplication.
Applying bounds computed in Proposition \ref{degbound} and Lemma \ref{degcurve}, and since $\tilde{n} \leq N$
and $\omega \leq 3$, we conclude that {\sf RatPar} runs within $\bigO(N^5 \theta(m,n,p)^2)$ arithmetic operations.
We conclude by recalling that for every $p=0, \ldots, r$, the routine {\sf RatPar} runs $\binom{m}{p}$ times.
\end{proof}

\section{Experiments}
\label{exper}

We present results of our tests on a Maple implementation of the algorithm {\sc SolveSDP}. We integrate this implementation
in the Maple library {\sc spectra}, {\it cf.} \cite{spectra}, whose main goal is to implement efficient exact algorithms for semidefinite
programming and related problems. The Version 1.0 of {\sc spectra} can be freely downloaded from the following web page:
\begin{center}
%\url{homepages.laas.fr/henrion/software/spectra}
\url{www.mathematik.tu-dortmund.de/sites/simone-naldi/software}
\end{center}
The rational parametrizations are computed using Gr\"obner bases via the Maple implementation of the software {\sc FGb}
\cite{faugere2010fgb}, exploiting the multilinearity of Lagrange systems already exhibited in Section \ref{degboundsoutputrep}
({\it cf.} \cite{FM11} for a tailored algorithm). The regularity assumptions on the input $(A,c)$ are also checked by testing
the emptiness of complex algebraic sets, hence performing Gr\"obner bases computations.

In Section \ref{randSDP} we use {\sc SolveSDP} to solve generic rank-con\-stra\-i\-ned semidefinite programs, giving details of
timings and output degrees of our implementations. In Section \ref{sosdec} we consider an application of our results for computing
certificates of nonnegativity for multivariate polynomials.

\subsection{Random SDP}
\label{randSDP}

In this test, we draw $(n+1)-$tuples of random $m \times m$ symmetric linear matrices $A_0,A_1, \ldots, A_n$ with rational coefficients.
The numerators and denominators of the rational entries are generated with respect to the uniform distribution in a given interval
(in our case, in $\ZZ \cap [-10^3,10^3]$). We also draw random linear forms $\ell_c = c^Tx$, and we consider different rank-constrained
semidefinite programs.
%%%% $\sdp_r$ with parameters $(m,n,r)$, with $r \in \{0, 1, \ldots, m-1\}$.

As explained in Section \ref{algorithm}, the most costly routine in {\sc SolveSDP} is the computation of rational parametrizations of
the Lagrange systems $\lag(\iota)$ defined in \eqref{LagSys}, namely Step 7 in the formal description in Section \ref{formalDescr}.
We report in Table \ref{tab:dense:sym} on timings (column {\sf SolveSDP}) and output degrees (column {\sf Deg}) relative to the
computation of the rational parametrization of a single Lagrange system. Ideally, we recall that to get the total time for {\sc SolveSDP}
one should take the sum of these timings for $p=0, \ldots, r$ weighted by $\binom{m}{p}$ (similarly to the complexity bound in Theorem
\ref{complexity}).

\begin{table}[!ht]
\centering
{\small
\begin{tabular}{|l|rr|l|rr|}
\hline
$(m,n,p)$ & {\sf SolveSDP} & {\sf Deg} & $(m,n,p)$ & {\sf SolveSDP} & {\sf Deg} \\
\hline
\hline
$(3,3,2)$ & 11 s   & 4   & $(5,3,3)$ & 3 s    & 20  \\
$(4,3,2)$ & 2 s    & 10  & $(5,4,3)$ & 1592 s & 90  \\
$(4,4,2)$ & 9 s    & 30  & $(5,5,3)$ & 16809 s & 207 \\
$(4,5,2)$ & 29 s   & 42  & $(5,2,4)$ & 7 s    & 20  \\
$(4,6,2)$ & 71 s   & 30  & $(5,3,4)$ & 42 s   & 40  \\
$(4,7,2)$ & 103 s  & 10  & $(5,4,4)$ & 42 s   & 40  \\
$(4,3,3)$ & 10 s   & 16  & $(5,5,4)$ & 858 s  & 16  \\
$(4,4,3)$ & 21 s   & 8   & $(6,6,3)$ & 704 s  & 112 \\
$(5,7,2)$ & 25856 s & 140 & $(6,3,5)$ & 591 s  & 80  \\
\hline
\end{tabular}
}
\caption{Optimization over $\calD_p \cap \RR^n$}
\label{tab:dense:sym}
\end{table}
We remark that our implementation is able to tackle from small to medium-size input semidefinite programs and different rank
constraints. As an example, for $(m,n,p)=(5,7,2)$ one should compute the critical points of a general linear form over
the algebraic set defined by $\binom{5}{3}\binom{5}{3}=100$ polynomials of degree $3$ in $7$ variables, which is unreachable
by the state-of-the-art algorithms: our implementation computes a rational parametrization of degree 140 after
seven hours. Further, when the size $m$ is fixed, the cost in terms of computation seems to reflect suitably
both the growth of output degree and of the number of variables $n$.

Moreover, it is worth to highlight that the entries of column {\sf Deg} coincide exactly with the {\it algebraic degree of SDP}
with parameters $(m,n,p)$, as computed in \cite[Table 2]{stu}. This fact is not obvious. Indeed, in \cite{stu} the algebraic
degree of SDP in rank $p$ (that is, on a solution of rank $p$) is understood as the degree of the complex variety $(\CC\calD_p)\chk{}$
dual to the variety $\CC\calD_p = \{x \in \CC^n \mymid \rank(A(x)) \leq p\}$. Our algorithm builds intermediate
incidence varieties whose degree is typically larger than the degree of the determinantal varieties and of their duals: hence one
could {\it a priori} expect the degree of the output representation to be larger than the expected degree (which si computed in \cite{stu}).
Even though the estimate of the output degree in Proposition
\ref{degbound} does not depend explicitly on formulas in \cite{stu}, but only on multilinear bounds, this fact is remarkable and
represents a guarantee of optimality of our method.

\subsection{Sum-Of-Squares certificates}
\label{sosdec}

In this final section, we consider an interesting application of rank-constrained semidefinite programming.
Let $u=(u_1, \ldots, u_k)$ and let $f \in \RR[u]_{2d}$ be a homogeneous polynomial of degree $2d$, for $d \geq 1$.
Let $b = \{\prod_i u_i^{j_i}\}_{\sum_i j_i=d}$ be the monomial basis of $\RR[u]_{d}$. The sum-of-squares (SOS)
decompositions of $f$ are parametrized by the so-called {\it Gram spectrahedron} of $f$:
\[
{\cal G}(f) = \{X \in \ss_{\binom{k+d-1}{d}}(\RR) \mymid X \succeq 0, \,\, f = b^T X b\},
\]
and any $X \in {\cal G}(f)$ is called a {\it Gram matrix} for $f$, {\it cf.} \cite{PowersWor}. Remark here that
the constraint $f = b^T X b$ is linear in the entries of $X$.
If $f=f_1^2+\cdots+f_r^2$, we say that $f$ has a SOS decomposition {\it of length} $r$. We deduce that deciding whether $f$ has a
SOS decomposition of length at most $r$ is equivalent to the following rank-constrained semidefinite program:
\begin{equation}
\label{rnkconslmi}
f = b^T X b \qquad X \succeq 0 \qquad \rank\,X \leq r.
\end{equation}
We have generated nonnegative polynomials by taking sums of squares of random homogeneous polynomials of degree $d$.
Applying {\sc SolveSDP} to this subfamily of problem $\sdp_r$, we have been able to handle example with $k \leq 3$ and
$2d \leq 6$, corresponding to Gram matrices of size 10. We believe that this is due to the particular sparsity of these
linear matrices. We give below direct examples of how the algorithm developed in this paper can be used in practice
to compute certificates of positivity for a given $f \in \RR[u]$.

\begin{example}[Chua, Plaumann, Sinn, Vinzant]
We consider the homogeneous binary sextic
\[
f = u_1^6-2u_1^5u_2+5u_1^4u_2^2-4u_1^3u_2^3+5u_1^2u_2^4-2u_1u_2^5+u_2^6 \in \RR[u_1,u_2]_6
\]
in \cite[Ex.\,4.4]{chuaetal}, and its Gram matrix
\[
A=
\begin{bmatrix}
1 & -1 & x_1 & -2-x_2 \\
-1 & -2x_1+5 & x_2 & x_3 \\
x_1 & x_2 & -2x_3+5 & -1 \\
-2-x_2 & x_3 & -1 & 1
\end{bmatrix}.
\]
Essentially by the Fundamental Theorem of Algebra, since $f$ is globally positive on $\RR^2$, we know
that it can be expressed as a sum of two squares. In a Maple worksheet, after the library {\sc spectra}
and the matrix $A(x)$ above has been entered, with the command
\begin{verbatim}
> SolveLMI(A,{rnk,deg,all},[2,3]);
\end{verbatim}
our library computes many solutions corresponding to different SOS-representations
of $f$. In particular, decompositions of length $2$ (minimal) and $3$, with information on the rank
of $A$ on every solution, and on the algebraic degree of its entries. It solves the rank-constrained
semidefinite program given in \eqref{rnkconslmi}. We give below the approximation to 20 decimal digits
of two SOS-representations, one of length $2$:
\[
\begin{array}{l}
x_1 \in \Big[-\frac{1617666671225218599972013}{604462909807314587353088}, -\frac{1617666671225218599972005}{604462909807314587353088}\Big] \approx -2.6762050160213870985 \\[.3em]
x_2 \in \Big[-\frac{3368250337925821251358839}{1208925819614629174706176}, -\frac{3368250337925821251358827}{1208925819614629174706176}\Big] \approx -2.7861513777574232861 \\[.3em]
x_3 \in \Big[-\frac{3235333342450437199944021}{1208925819614629174706176}, -\frac{3235333342450437199944017}{1208925819614629174706176}\Big] \approx -2.6762050160213870985 \\[.3em]
\end{array}
\]
and one of length $3$:
\[
\begin{array}{l}
x_1 \in \Big[\frac{3203539382882212253342931}{2417851639229258349412352}, \frac{1601769691441106126671543}{1208925819614629174706176}\Big] \approx 1.3249528345351282960 \\[.3em]
x_2 \in \Big[-\frac{2700826142354717756217093}{2417851639229258349412352}, -\frac{1350413071177358878108513}{1208925819614629174706176}\Big] \approx -1.1170355114161030782 \\[.3em]
x_3 \in \Big[\frac{1696463549117506376965235}{1208925819614629174706176}, \frac{3392927098235012753930515}{2417851639229258349412352}\Big] \approx 1.4032817577329022769 \\[.3em]
\end{array}
\]
In addition, some {\it rational} SOS-representations are computed, such as
\[
\begin{array}{lcccl}
x_1 \in [0,0]   & &            & & x_1 \in [2,2] \\
x_2 \in [-2,-2] & & \text{and} & & x_2 \in [-2,-2] \\
x_3 \in [2,2]   & &            & & x_3 \in [0,0] \\
\end{array}
\]
Finally, the following rational parametrization defines a finite set (of $4$ elements) containing
one point where the matrix $A(x)$ is positive semidefinite and has rank 2:
\begin{align*}
q(t) & = t^4-2t^3-5t^2+16t-11 \\
q_0(t) & = 20t^9-180t^8+576t^7-448t^6  - 1917t^5  + 6130 t^4  - 8058 t^3  + 5475 t^2  - 1787 t + 187  \\
q_1(t) & = 20 t^9  - 156 t^8  + 284 t^7  + 1070 t^6  - 6294 t^5  + 13725 t^4  - 16087 t^3  + 10434 t^2  - 3371 t+ 374 \\
q_2(t) & = -36 t^9  + 330t^8  - 1116 t^7  + 1233 t^6  + 2230 t^5  - 9040 t^4  + 12678 t^3  - 9040 t^2  + 3138 t - 374 \\
q_3(t) & = 20 t^9  - 144 t^8  + 192 t^7  + 1278 t^6  - 6130 t^5  + 12087 t^4  - 12775 t^3  + 7148 t^2  - 1683 t.
\end{align*}

\end{example}

\begin{example}
The following ternary quartic
\[
f = u_1^4+u_1u_2^3+u_2^4-3u_1^2u_2u_3-4u_1u_2^2u_3+2u_1^2u_3^2+u_1u_3^3+u_2u_3^3+u_3^4.
\]
is a sum of two squares, while the general nonnegative ternary quartic is a sum of three squares.
This degeneracy can be checked by our algorithm. The Gram matrix of $f$ is a $6 \times 6$ linear
symmetric matrix in $6$ variables $x_1,\ldots,x_6$. The exact representation of the nonnegativity
certificate for $f$ is then given by the following representation:
{\small
\[
x_1 = \frac{3+16t}{-8+24t^2} \qquad x_2 = \frac{8-24t^2}{-8+24t^2} \qquad x_3  = \frac{8+6t+8t^2}{-8+24t^2}
\]
\[
x_4 = \frac{16+6t-16t^2}{-8+24t^2} \qquad x_5 = \frac{-3-16t}{-8+24t^2} \qquad x_6 = \frac{3+16t}{-8+24t^2}
\]
}
where $t$ is one of the roots of $q(t) = t^3-t-1$. The corresponding Gram matrix has rank 2.
\end{example}

\section{Final remarks}
This paper addresses a fundamental problem in computational real algebraic geometry, that is rank-constrained semidefinite programming.
Our algorithm is able to return an exact algebraic representation of all minimizers, with explicit bounds on its output degree and whose
complexity is essentially quadratic on the mentioned degree bound. The algorithm works under assumptions on the input, which are proved
to be generically satisfied. This is done by exploiting the determinantal structure of this optimization problem, and by reducing it to
linear optimization over determinantal varieties. This reduction step allows to manage (non-convex) additional rank constraints.
To the best of our knowledge, this is the first exact algorithm for solving $\sdp_r$.

\section*{Acknowledgements}

The author thanks the organizers of the Thematic Program on Computer Algebra, held at the Fields Institute, Toronto, Canada
from July to December 2015, where this paper was prepared. He thanks in particular E. Kaltofen and \'E. Schost for helpful
discussions about the topic of the paper. Finally, the author thanks the anonymous reviewers for having improved the first
version of the paper \cite{naldiIssac2016} published in the Proceedings of ISSAC 2016, of which this paper represents the
extended version.

\bibliographystyle{elsarticle-harv}
\bibliography{../../../../BIB/biblio}{}

\end{document}